\RequirePackage{fix-cm}
\documentclass{arxive}

\usepackage[T1]{fontenc}
\usepackage[utf8]{inputenc}
\usepackage{amsmath,amssymb,amsthm,mathtools}
\usepackage{bm}
\usepackage{bbm}
\usepackage{braket}

\usepackage{graphicx}
\usepackage{booktabs}
\usepackage{tikz}
\usetikzlibrary{arrows.meta,positioning,calc}
\usepackage[square, sort&compress, numbers]{natbib}

\usepackage[dvipsnames]{xcolor}

\newtheorem{theorem}{Theorem}
\newtheorem{lemma}{Lemma}
\newtheorem{remark}{Remark}
\newtheorem{definition}{Definition}
\newtheorem{proposition}{Proposition}
\newtheorem{corollary}{Corollary}

\newcommand{\HALT}{\mathrm{HALT}_{\mathcal U}}
\newcommand{\runin}[1]{\par\noindent\textit{#1.}\ }

\begin{document}

\title{{Constructive realization of self-referential prediction limits in quantum control: Resource bounds and Gödel-safe architectures}}

\author{Salman Sajad Wani$^{1,*}$, \'Alvaro Perales-Eceiza$^2$, Saif Al-Kuwari$^1$, Mir Faizal$^{3,4,5,6}$}

\affil{$^1$Qatar Center for Quantum Computing, College of Science and Engineering, Hamad Bin Khalifa University, Doha, Qatar}

\affil{$^2${Computer Engineering Department, Polytechnic School}, Universidad de Alcal\'a, 28805 Madrid, Spain}

\affil{$^3$Canadian Quantum Research Center, 204-3002 32 Ave, Vernon, BC V1T 2L7, Canada}
\affil{$^4$Irving K.\ Barber School of Arts and Sciences, University of British Columbia Okanagan, Kelowna, BC V1V 1V7, Canada}
\affil{$^5$Department of Mathematical Sciences, Durham University, Upper Mountjoy, Stockton Road, Durham DH1 3LE, UK}
\affil{$^6$Faculty of Sciences, Hasselt University, Agoralaan Gebouw D, Diepenbeek, 3590, Belgium}

\noindent\textbf{Corresponding author email:} salmansajadwani@gmail.com

\begin{abstract}
Programmable quantum control systems increasingly rely on predictive modules for certification, real-time feedback, and autonomous decision-making{. This development raises} a fundamental question: can self-analyzing quantum platforms universally predict their own experimental outcomes? {Wolpert formalized a general impossibility of universal self-prediction{. Here} we translate that limitation into an explicit {laboratory obstruction that can be realized with finite resources. We consider settings with programmable quantum control} in which} predictors can be embedded as subroutines within the experiments they analyze. {Our diagonal construction uses} Kleene's recursion theorem {to transform} any deterministic bounded-time predictor into a reversible protocol encoding its own specification. {The resulting} protocol invokes the predictor {on that specification} and deterministically {produces} a classical pointer record {that contradicts} the forecast. {For efficient predictors, the} compilation has polynomial overhead and admits concrete physical realizations as a fault-tolerant quantum circuit and as a minimal {Mach--Zehnder} interferometer{.} {{These realizations connect}} computability-theoretic {self-reference} to programmable quantum hardware. {We also introduce and formally define G\"odel-safe {architectures. These architectures block the forbidden causal path from the protocol description to an actuator that can affect the pointer during the same run. We} analyze their implications for real-time quantum error correction, including the resulting expressiveness trade-offs. As quantum control loops grow in computational expressiveness, the limits of self-reference cease to be mere mathematical abstractions and become explicit engineering constraints for the reliable operation of autonomous quantum technologies.}
\end{abstract}


\section{INTRODUCTION}

The increasing sophistication of quantum control systems has enabled real-time certification, autonomous feedback, and predictive decision-making in programmable platforms. As these systems evolve toward self-analyzing architectures{,} {their predictive modules can} operate on representations of their own control protocols{. This development raises} a fundamental question: what are the limits of algorithmic prediction when the predictor can be embedded within the system it analyzes? {The theoretical impossibility of universal self-prediction is already established for abstract inference devices. D. Wolpert formalized this limitation through diagonalization, demonstrating that any finite inference device can be placed in a physical situation forcing contradiction with its own prediction \cite{Wolpert2008}. In our experimental context, ``diagonalization''\footnote{{Here} ``diagonalization'' {denotes a computability-theoretic construction; it is unrelated to diagonalizing a matrix or Hamiltonian. The term} originates in Cantor's diagonal argument (1891) {and related} methods later used by G\"{o}del (1931) to {prove} incompleteness and by Turing (1936) to establish the undecidability of the halting problem.} {has a simple operational meaning. The predictor is asked which bit the experiment will record. Its forecast becomes available to the controller during that same run. The controller then chooses a binary setting---or, in the optical realization, a phase---that makes the experiment record the opposite bit. Consequently, no deterministic, deadline-bounded predictor can be correct on every protocol in a class that permits this causal loop within the same run.}} {In the ideal Mach--Zehnder realization, the predictor forecasts which detector will click, and the controller chooses \(\phi=0\) or \(\pi\) so that the other detector clicks.}

{{This paper gives} a constructive realization of {Wolpert's intentionally abstract limit. It converts the limit} into an operational {laboratory obstruction implemented within a prescribed deadline}.} Concretely, we ask whether a deterministic {clocked} procedure {can correctly predict, by a stated deadline, the outcome recorded in a single shot.} {{The procedure has a} declared worst-case runtime bound {and is given}} a finite classical specification of {the experiment---its} control logic, settings, and readout rule. {Using Kleene's recursion theorem \cite{Rogers1967},} we prove that no such procedure can be universally correct on any {protocol class satisfying the reflective closure condition}. {For any fixed predictor $A$, we construct an explicit compiler $A \mapsto \mathcal{C}_A${. The compiler produces} a logically reversible protocol whose controller evaluates $A$ on {the protocol's own finite classical description. The controller then} deterministically enforces a recorded outcome that disagrees with $A$. The compilation incurs only polynomial overhead in logical depth and width{. It therefore admits a} standard fault-tolerant realization. Bennett uncomputation \cite{Bennett1973,Bennett1989} ensures the contradiction is exposed as a clean pointer bit rather than hidden in computational history. A minimal optics instantiation, a Mach--Zehnder interferometer {with a controlled phase}, makes the obstruction concrete using standard experimental components. Our result is logically distinct from prior results establishing undecidability in physical systems \cite{Eisert2012PRL,Cubitt2015,ShiraishiMatsumoto2021Undecidability,PeralesEceiza2025PhysRep}{. Those results} predominantly concern asymptotic or thermodynamic-limit properties{.}} {{By contrast,} the obstruction here is instance-wise, deadline-bounded, and realized within a prescribed runtime.}

Our construction connects three domains: computability theory (Kleene's fixed points), reversible computation (Bennett's uncomputation), and quantum mechanics (unitary evolution followed by {measurement in the computational basis}). The diagonal obstruction is computational {rather than} quantum-mechanical: it arises from self-reference, not quantum indeterminacy. {Measurement occurs in the computational basis after uncomputation has disentangled the pointer from the computational history, yielding a deterministic outcome.} {{Classical} reversible computation suffices in principle to realize the required reversible dynamics \cite{Bennett1973}{. Quantum circuits nevertheless} provide the most natural implementation framework{. Their} gates are inherently reversible, {fault tolerance} is rigorously established, and the tight integration of predictive logic and pulse routing in quantum control stacks creates the physical conditions for the self-referential vulnerability.}

{The {reflective} closure condition underlying the theorem is a precise architectural hypothesis rather than a generic property of quantum-control platforms. Modern quantum systems already contain several components of {such self-analyzing architectures}, including stored protocol descriptions, programmable low-latency control, embedded analysis routines, and feed-forward control decisions. However, the reflective closure condition required for the theorem is stronger{. The} current protocol description must be available to an analysis module as input data{. The} module output must {also} be routable, within the same run, into a subsequent {control setting that can affect the pointer. We call this route from the protocol description, through the analysis module, to pointer control within the same run the \emph{forbidden causal path}. The reflective closure condition is the availability of that path within the protocol class, and the diagonal construction exploits it.} This identifies a concrete architectural boundary for self-analyzing quantum technologies. The resulting practical issue is architectural: one must specify which data dependencies are allowed in self-analyzing controllers. {A universal prediction guarantee for the current run requires the forbidden causal path to be blocked. The current} protocol description may {still} be logged, hashed, audited, or {used to configure later runs. It must not, however, be} routed to {an actuator that can affect the pointer during the current run}.}

{{The obstruction therefore has concrete implications} for contemporary engineering practice. We first evaluate the {reflective closure condition} across modern control hardware {and show that the forbidden causal path} is physically admissible on current platforms. {Finite} memory, instruction sets, and bus architectures {support this capability, but do not enable it} by default. To avoid the deterministic failure forced by self-reference, we formally propose and define the \textit{G\"{o}del-safe architecture}{. This architecture enforces the \textit{G\"{o}del-safe condition} by blocking the forbidden causal path. The design logic parallels compiler theory: specific guarantees become verifiable when problematic control flows are forbidden, at the cost of reduced expressiveness; Sec.~\ref{Sec:computabilty} develops this analogy.} We then apply {this condition} to real-time Quantum Error Correction (QEC){. This application characterizes} the expressiveness trade-offs that arise when predictive modules are isolated from their own protocol descriptions {and discusses} their implications for broader autonomous quantum-control settings.}

{The remainder of the paper is organized as follows. Section~\ref{sec:Methods} introduces the constructive framework and the diagonal {construction}. Section~\ref{sec:Results} presents the theorem, proof, and experimental proposal. Section~\ref{sec:Discussion} analyzes the architectural implications, including G\"odel-safe {architectures} and applications to real-time QEC. Section~\ref{sec:Conclusions} summarizes the main conclusions and future directions. Detailed proofs, resource bounds, and noise modeling are provided in the Supplemental Material{.}}
\section{METHODS}
\label{sec:Methods}

We {combine} constructive methods from computability theory
({specifically, }Kleene's recursion theorem {and the $s-m-n$ theorem})
with reversible quantum computation{. This section specifies finite protocol
descriptions and clocked predictors. It then identifies the routing capability
required during the same run and captured by the reflective closure condition,
and constructs the compiler used in Theorem~\ref{thm:noprediction-letter}.}

\subsection{Diagonal construction}

We {formulate the diagonal construction for forecasting a single-shot outcome
by a prescribed deadline} in programmable closed-loop experiments. Each laboratory
procedure {$\mathcal{C}$} is assumed to admit a finite classical {binary encoding,
denoted $\xi={\mathcal{C}}\in\{0,1\}^*$} (netlist, pulse program, and
readout convention). A
clocked predictor is a fixed deterministic module $A$ 
that, on any input $\xi$, outputs {an} $m$-bit string
and halts by its declared deadline $T_A$. {The relevant control model is therefore finite-protocol based, rather than a
single microscopic Hamiltonian specification. The
control data are the classical settings encoded in the protocol description
$\xi$: gate choices, pulse or phase commands, feed-forward bits, timing
specifications, and readout conventions. The diagonal construction requires
only the following routing capability: within the same run, the first output
bit of $A(\xi)$ must be routable to a binary {setting that can affect the pointer} before
terminal readout. Thus, the control layer has the causal form
\begin{equation}
\label{eq:protocol-control-layer}
\xi
\longrightarrow
A(\xi)
\longrightarrow
\text{binary feed-forward setting}
\longrightarrow
X .
\end{equation}
A microscopic Hamiltonian may specify the platform-specific implementation of
a given control setting. The no-prediction argument depends on the finite
protocol description and this causal control structure.}

For every such $A$, we construct an explicit compiler $A\mapsto \mathcal{C}_A$ producing a finitely specified, logically reversible experiment. This protocol $\mathcal{C}_A$ contains an internally accessible copy of {${\mathcal{C}_A}$} and a designated pointer bit $X$ in the recorded classical string $Y_A$.

{The compilation process relies on two fundamental results from computability theory{. These results} ensure that the construction is computable and finitely specified. First, the $s$-$m$-$n$ theorem guarantees that the injection of a fixed parameter (in this case, the protocol's own classical description) into the control sequence is a computable operation with finite overhead. Building on this parameterization, Kleene's second recursion theorem establishes the existence of the self-referential fixed point \cite{Rogers1967}. Together, {the two results} guarantee that such a protocol $\mathcal{C}_A$ can be effectively constructed{. The protocol can evaluate} the predictor $A$ on its own complete description {${\mathcal{C}_A}$}
\footnote{Formally, \textbf{the $s$-$m$-$n$ theorem} states that for $m$ fixed parameters and $n$ free variables, a program description can be computably transformed so as to hardcode the fixed parameters. In the $m=1$, $n=1$ case, given a program description $e$ for a partial computable function $f(x,y)$, there exists a total computable function $s_1^1(e,x)$ that outputs the description of a program computing the unary function $y \mapsto f(x,y)$. \textbf{Kleene's second recursion theorem} states that for any total computable function $h$ mapping program descriptions to program descriptions, there exists a program description $e$ such that $e$ and $h(e)$ compute the same partial computable function.}.} {The pointer relation enforced within the prescribed deadline and its no-prediction consequence are stated formally in Theorem~\ref{thm:noprediction-letter} and verified in Sec.~\ref{sec:SMIII} of the Supplemental Material.}

\subsection{Physical realizability}

The construction is operational: $\mathcal{C}_A$ can be realized in finite
laboratory time{. Its reversible control uses} Bennett uncomputation
\cite{Bennett1973,Bennett1989} {and admits} unitary embedding up to the final
pointer readout{. That readout is} the only irreversible step. {The}
compiler {therefore outputs a} reversible control protocol with {a finite schedule,}
an explicit deadline{, and an} auditable classical pointer.

{Sections~\ref{sec:mzi-letter} and \ref{sec:ft-resource-scaling} develop two
physical realizations. A} {Mach--Zehnder} interferometer {with a controlled phase
implements the final binary pointer}. {A} fault-tolerant quantum circuit {provides
the second realization using} standard gates with polynomial overhead{. These
implementations make} the construction accessible to current noisy intermediate-scale
quantum (NISQ) platforms. {Operational admissibility requires the protocol language to} invoke
$A$ as a subroutine and route its output into {an} ordinary feed-forward {setting
during the same run; this is the capability captured by the reflective closure
condition. The architectural consequences of blocking this path are analyzed
in Sec.~\ref{sec:Discussion}.}

\section{RESULTS}
\label{sec:Results}

We now state and prove our main no-prediction theorem for clocked deterministic predictors.

\begin{theorem}[No-prediction for clocked deterministic predictors]
\label{thm:noprediction-letter}
Fix any deterministic clocked predictor $A$ of code length $k:=|{A}|$ with a declared worst-case runtime bound $\tau(k)$, meaning that for every admissible description $\xi$, the computation $A(\xi)$ halts within $\tau(k)$ steps and outputs an $m$-bit string. Let $P_A^{(0)}(\xi)\in\{0,1\}$ denote the first output bit of $A(\xi)$.

Consider any experiment class whose protocol language permits, within a single run, (i) executing $A$ as a subroutine on a supplied finite description $\xi$, and (ii) routing the resulting bit $P_A^{(0)}(\xi)$ into an ordinary binary control setting prior to the terminal pointer readout. Then there exists a finitely specified, logically reversible (hence unitarily embeddable up to the final pointer readout) experiment $\mathcal C_A$ in that class which produces an $m$-bit classical record $Y_A$ with first bit $X\in\{0,1\}$ satisfying
\begin{equation}
\label{eq:diag-letter}
X \;=\; 1 - P_A^{(0)}\!\bigl({\mathcal C_A}\bigr).
\end{equation}
In particular, the first-bit forecast {$P_A^{(0)}({\mathcal C_A})$} is wrong on that very instance, and hence {$A({\mathcal C_A})\neq Y_A$}.

The instance $\mathcal C_A$ is obtained by an explicit compilation $A\mapsto \mathcal C_A$; the required self-application {${\mathcal C_A}$} is guaranteed by Kleene's recursion theorem~\cite{Rogers1967}. If $\tau(k)$ is polynomial in $k$, then the compiled protocol has polynomial logical depth and width (explicit resource bounds and fault-tolerant embeddings are given in the Supplemental Material).
\end{theorem}

\subsection{Proof roadmap}
The proof {proceeds in three steps. First, Kleene's} recursion theorem {supplies a
fixed point that makes the compiled description available to the controller.
Second, a bounded reversible computation obtains the forecast before the
scheduled deadline and writes its complement to the pointer. Third, Bennett
uncomputation clears the work registers before terminal readout. The formal
model and supporting constructions are given in Supplemental
Secs.~\ref{sec:SMI}--\ref{sec:SMIII}, particularly
Lemmas~\ref{lem:SMII-FY}, \ref{lem:SMIII-DA}, and
\ref{lem:SMIII-contradict}; the next subsection explains how these steps
produce the clean pointer record.}

\subsection{Construction and intuition} 
The instance $\mathcal C_A$ is a {feed-forward} loop {produced by the diagonal compiler}. Using the fixed-point guarantee of Kleene's recursion theorem~\cite{Rogers1967}, the controller in $\mathcal C_A$ can supply $A$ with the  {actual} compiled description ${\mathcal C_A}$, compute the forecast bit $b:=P_A^{(0)}({\mathcal C_A})$ by a hard deadline $t_0$, and then actuate the complementary pointer value $X:=1-b$. This enforces Eq.~\eqref{eq:diag-letter} on that very run and therefore the instance-wise failure $A({\mathcal C_A})\neq Y_A$. The runtime bound $\tau(k)$ is what makes the loop operational rather than evasive: it fixes a finite schedule in which the forecast must be produced strictly before the actuation time. Finally, the controller is uncomputed after writing the pointer, leaving a clean pointer subsystem for terminal readout.

{Operationally, there are only two cases: if the predictor returns \(b=0\), the controller sets \(X=1\); if it returns \(b=1\), the controller sets \(X=0\). Kleene's recursion theorem supplies the fixed point required to evaluate the predictor on the compiled protocol's own description; the complementation \(X=1-b\) itself is elementary.}

\begin{figure}[t]
  \includegraphics[width=\linewidth]{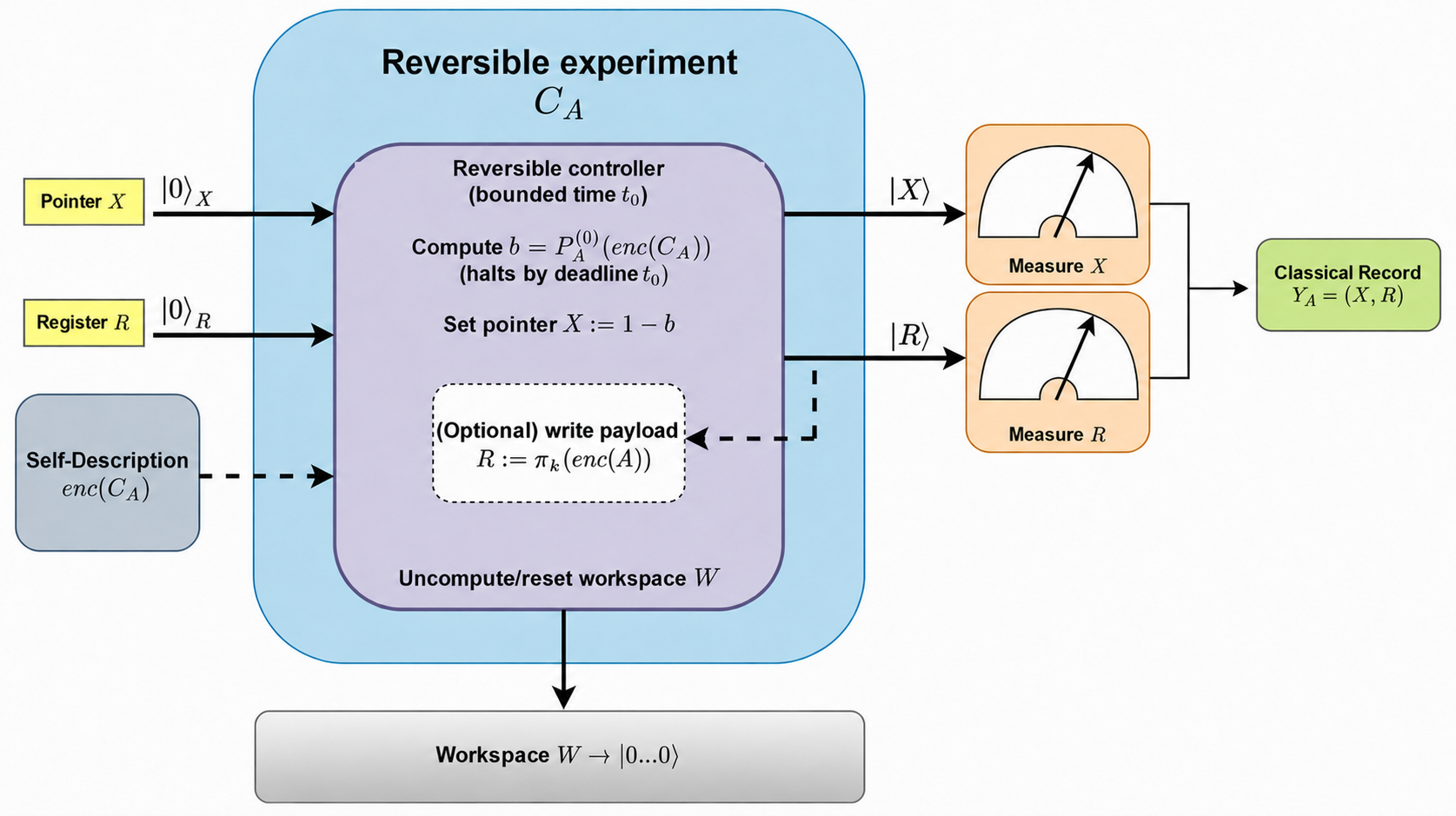}
\caption{\label{fig:control-loop}Compiled control loop for $\mathcal C_A$. A reversible controller computes $b=P_A^{(0)}({\mathcal C_A})$ by the deadline $t_0$ and sets the pointer bit $X=1-b$. Optionally, a payload $R=\pi_k({A})$ is written to a second register. The workspace is uncomputed before the final measurement {in the computational basis}, yielding the record $Y_A=(X,R)$ {for a single shot}.}
\end{figure}

For an information-preserving variant, one may take $m:=k{+}1$ and record $Y_A=(X,R)$, where $R:=\pi_k({A})\in\{0,1\}^k$ for a fixed length-preserving bijection $\pi_k$ with computable inverse. The no-prediction conclusion depends only on $X$; the payload simply makes the observed record explicitly information-conserving.

{
We decompose $\mathcal C_A$ into three reversible stages. Here, ``reversible''
is meant in the standard sense of reversible computation: each stage is a
bijection on the full logical configuration space, including clock registers,
work registers, ancillas, pointer registers, and payload registers, when
present. Consequently, if $F$ denotes the classical reversible map implemented
by a given stage, its quantum embedding is the permutation unitary
\begin{equation}
\label{eq:permutation-unitary}
U_F\ket{z}=\ket{F(z)}
\end{equation}
on computational-basis states, extended linearly to the full Hilbert
space{~\cite{Feynman1982}}.
The three stages are therefore unitary circuit blocks up to terminal readout.
The sole irreversible operation is the final computational-basis measurement
of the isolated pointer registers. 
In particular, for a clean pointer register initialized in $\ket{0}_X$, the
assignment $X:=1-b$ is implemented reversibly as
\begin{equation}
    \ket{b}\ket{0}_X
\longmapsto
\ket{b}\ket{1-b}_X.
\end{equation}
after which Bennett uncomputation reverses the internal computation while
preserving the pointer record for measurement~\cite{Bennett1973,Bennett1989}.
The three stages are:
}
\begin{enumerate}
\item[(D)]  {Deadline-bounded forecast.}
Using the fixed point from Ref.~\cite{Rogers1967}, the controller supplies $A$ with ${\mathcal C_A}$ and computes $b:=P_A^{(0)}({\mathcal C_A})$. Since $A$ halts within $\tau(k)$, the schedule sets a hard deadline
$t_0=\tau(k)$ plus fixed compilation/marshalling overhead (see Supplemental Material),
and then writes $X:=1-b$ into a clean pointer register.

\item[(P)]  {Payload writer (optional).}
Write $R:=\pi_k({A})$ into a clean register by reversible bit operations.

\item[(U)]  {Global uncomputation.}
Reverse the controller’s internal computation, returning clocks and ancillas to a standard blank state. Immediately before measurement, the pointer registers $(X,R)$ are isolated from the rest of the device, so the terminal readout is an ordinary computational-basis measurement of a classical record~\cite{Bennett1989}.
\end{enumerate}

\subsection{Experimental proposal}
\label{sec:mzi-letter}

{
The diagonal switch admits a direct laboratory realization as a {binary pointer implemented by} a single-photon Mach--Zehnder interferometer (MZI) {with a controlled phase}.  {Figure~\ref{fig:mzi-diagonal-setup} summarizes the complete implementation
chain, distinguishing the classical controller layer from the optical pointer
layer: the controller evaluates $A$ on the compiled description
${\mathcal C_A}$, computes the diagonal bit, and issues the corresponding
phase command to the MZI.}  Here, the MZI functions as an interferometric pointer by converting a discrete phase command into {a detection event at an output port through interference of a single photon}. A classical control bit selects $\phi\in\{0,\pi\}$; in the ideal limit, the interferometer maps this phase to one of the two output ports, and the corresponding detector event yields the pointer bit. The diagonal step is therefore a clocked feed-forward cycle in which the controller computes a bit, sets the optical phase, and records one classical bit in a single shot.
}

Consider an MZI with balanced $50{:}50$ beam splitters BS$_1$, BS$_2$, a phase shifter on one arm, and threshold detectors $D_0,D_1$ at the outputs. The path degree of freedom is a qubit with basis $\{\ket{0},\ket{1}\}$ (lower/upper arm), with the phase shifter acting on $\ket{1}$. {
The primary physical control parameter in this realization is the
interferometric phase $\phi$. At the path-qubit level, the corresponding
phase operation is
\begin{equation}
\label{eq:mzi-phase-unitary}
U_\phi
=
\ket{0}\!\bra{0}
+
e^{i\phi}\ket{1}\!\bra{1}.
\end{equation}
The diagonal controller restricts the phase to the two settings
\begin{equation}
\label{eq:mzi-binary-phase-control}
\phi=\pi b_{{\mathrm{diag}}},
\qquad
\phi\in\{0,\pi\}.
\end{equation}
} Under the standard convention for BS$_1$, a single-photon input prepares $\ket{\psi_{\mathrm{path}}}=(\ket{0}+\ket{1})/\sqrt 2$, and after a relative phase $\phi$ on $\ket{1}$ and recombination at BS$_2$, the ideal click probabilities are
$\Pr(D_0\mid\phi)=\cos^2(\phi/2)$ and $\Pr(D_1\mid\phi)=\sin^2(\phi/2)$ \cite{BornWolf,SalehTeich}.
We encode $X=0$ for a click at $D_0$ and $X=1$ for a click at $D_1$; thus $\phi=0$ yields $X=0$ and $\phi=\pi$ yields $X=1$ in the ideal limit. In this sense, the MZI provides a binary pointer whose value is fixed by a one-bit classical control input.  {The physical mechanism is phase-to-port conversion: BS$_1$ prepares a
coherent path superposition; the phase shifter imposes a relative phase between
the two arms; and BS$_2$ converts the relative phase into constructive
interference at one output and destructive interference at the
other~\cite{BornWolf,SalehTeich}.}

{
Let $A$ be the bounded-time predictor as defined earlier. The controller runs
$A$ on the actual compiled description ${\mathcal C_A}$ to obtain
\begin{equation}
b:=P_A^{(0)}({\mathcal C_A}).
\end{equation}
{At the protocol level, the} feed-forward control bit in the diagonal construction is
\begin{equation}
b_{\mathrm{diag}}:=1-b.
\end{equation}
In the MZI realization, this bit is converted into the physical phase control
of the interferometer:
\begin{equation}
\phi=\pi\,b_{\mathrm{diag}},
\qquad
\phi\in\{0,\pi\}.
\end{equation}
In the ideal limit this yields the single-shot diagonal relation
\begin{equation}
X=b_{\mathrm{diag}}
=
1-P_A^{(0)}({\mathcal C_A}),
\end{equation}
so one detector click implements the diagonal flip as an ordinary classical
measurement record. If used, payload bits are written to a separate classical
register; the MZI supplies only the designated pointer bit $X$.
}

To exclude race conditions, the commanded phase must settle before the photon completes the controlled arm. Let $T_{\mathrm{ctrl}}$ be the worst-case time to marshal ${\mathcal C_A}$, run $A$ to obtain $b$, and compute $b_{\mathrm{diag}}=1-b$, including interface latency, and let $T_{\pi}$ be the modulator settling time for a $\{0,\pi\}$ command. The optical layout must satisfy $T_{\mathrm{flight}}\ge T_{\mathrm{ctrl}}+T_{\pi}+\Delta$, where $T_{\mathrm{flight}}$ is the photon time of flight through the controlled arm (including any delay line) and $\Delta>0$ is a safety margin. The clocked assumption makes this checkable: since $A$ halts within $\tau(k)$, one upper-bounds $T_{\mathrm{ctrl}}$ in terms of $k$ and fixed compilation overheads (Supplemental Material) and chooses the delay accordingly. {The quantities $T_{\rm ctrl}$, $T_\pi$, $T_{\rm flight}$, and
$\Delta$ are scheduling parameters. Their role is to enforce causal timing in
the diagonal construction: they ensure that the selected phase in
Eq.~\eqref{eq:mzi-binary-phase-control} is applied before the photon reaches
the phase modulator.}

We summarize robustness by a compact {error envelope for a single shot} (noise model and derivation in the Supplemental Material). Let $V\in[0,1]$ be the fringe visibility, $\delta\phi$ the RMS phase-setting error about the targets $\{0,\pi\}$, $\eta$ the {overall efficiency for detecting a single photon from source to detector}, and $d_{\mathrm{bg}}\ll 1$ the per-detector background click probability per coincidence gate. We postselect on $\mathrm{acc}$, the event that exactly one detector clicks within the gate (discarding double clicks). For matched detectors (equal efficiencies and background rates), the postselected bit error obeys, to leading order in $d_{\mathrm{bg}}$,
$\Pr[X\neq b_{\mathrm{diag}}\mid \mathrm{acc}] \le \bigl((1-V)/2+\sin^2(\delta\phi/2)\bigr) + ((1-\eta)/\eta)\,d_{\mathrm{bg}}$.
The first contribution bounds wrong-port clicks from imperfect interference (finite $V$) and phase jitter; the second captures the dominant loss--background mechanism in the accepted sample, namely a missed detection together with a single spurious click in the wrong port.  {The quantities $V$, $\delta\phi$, $\eta$, and $d_{\rm bg}$ are
calibration and noise parameters of the optical implementation. They quantify
the reliability of the pointer relation {$X=b_{\mathrm{diag}}$} and are distinct
from the primary controls entering the diagonal construction.}

{
\begin{figure}[t]
    \centering
    \includegraphics[width=\linewidth]{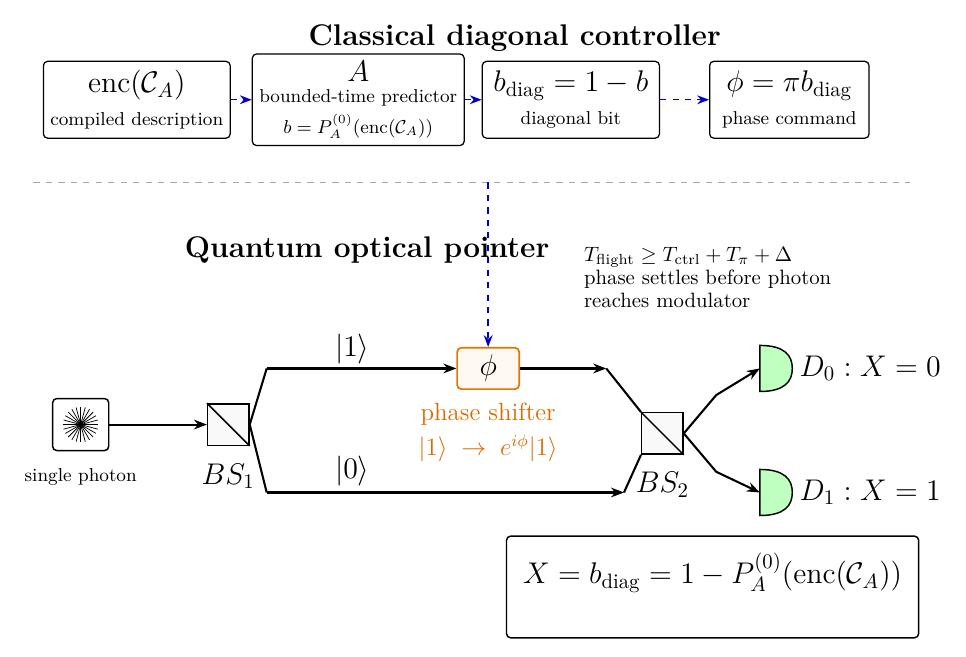}
\caption{\label{fig:mzi-diagonal-setup}
{Single-photon Mach--Zehnder implementation of the diagonal pointer bit.
The upper layer depicts the classical diagonal controller: the compiled
description ${\mathcal C_A}$ is supplied to the bounded-time predictor
$A$, whose first output bit $b=P_A^{(0)}({\mathcal C_A})$ is
complemented to form {$b_{\mathrm{diag}}=1-b$}. The commanded phase is
{$\phi=\pi b_{\mathrm{diag}}$}. The lower layer depicts the optical pointer: a
single photon enters a balanced MZI, the phase shifter acts on the upper path
$\ket{1}$ as $\ket{1}\mapsto e^{i\phi}\ket{1}$, and the output detectors
define the pointer assignments $D_0\mapsto X=0$ and $D_1\mapsto X=1$. In
the ideal limit, the detector record satisfies
{$X=b_{\mathrm{diag}}=1-P_A^{(0)}({\mathcal C_A})$}. The timing condition
$T_{\rm flight}\ge T_{\rm ctrl}+T_\pi+\Delta$ ensures that the phase settles
before the photon reaches the phase modulator.}
}
\end{figure}
}

\subsection{Fault-tolerant resource scaling}
\label{sec:ft-resource-scaling}
To keep the diagonal construction operational (rather than purely formal), we track finite time and space under an explicit commitment deadline. {The compiled protocol has a finite cutoff.} For a predictor $A$ of description length $k:=|{A}|$ and time-constructible budget $\tau(k)$, {we define this cutoff as $t_0 := {g_{\mathrm{comp}}}\!\bigl(k,|x^\ast|\bigr)+\tau(k)+1$. Here,} $x^\ast={D_A}$ is the recursion-theorem fixed point implementing the diagonal experiment specialized to $A$, and {$g_{\mathrm{comp}}$} denotes the (time-constructible, polynomial) overhead of compilation and orchestration. The fixed point does not introduce asymptotic overhead from self-reference{. For} any fixed diagonalization pipeline{,} there exists a constant $c_{\mathrm{fp}}$, independent of $A$, such that $|x^\ast|\le k+c_{\mathrm{fp}}$ and {$t_0\le \tau(k)+{\bar g_{\mathrm{comp}}}(k)+1$, with ${\bar g_{\mathrm{comp}}}(k):={g_{\mathrm{comp}}}(k,k+c_{\mathrm{fp}})$} (see Supplemental Material). In particular, polynomial $\tau(k)$ implies polynomial $t_0$.

A conservative clocked implementation {makes the obstruction explicit at polynomial logical cost. Specifically, it yields} $D_{\log}\lesssim t_0\log t_0+O(k)$ and $W_{\log}\lesssim t_0+O(k)$, plus $k{+}1$ pointer qubits. We translate these logical bounds into a fault-tolerant reliability statement under standard local stochastic noise. {For a surface-code compilation} with physical two-qubit error rate $p<p_{\mathrm{th}}$, the logical failure per location is exponentially suppressed with code distance $d$, namely $p_L(d)\le c_0(p/p_{\mathrm{th}})^{(d+1)/2}${. A} conservative union bound {then} gives $P_{\mathrm{logical}}\le N_{\mathrm{loc}}\,p_L(d)$ with $N_{\mathrm{loc}}\sim \alpha\,D_{\log}W_{\log}$. {Choose} $d$ to scale linearly with $k$ (e.g.\ $d=3k+1$){. For} $p<p_{\mathrm{th}}$ and polynomial $\tau(k)${, so that} $N_{\mathrm{loc}}=\mathrm{poly}(k)${, the exponential suppression then dominates the polynomial growth of $N_{\mathrm{loc}}$. The} overall logical failure {therefore} decays exponentially in $k$.

{Within this compilation model, the {overhead in physical qubits and the timing overhead for the} surface-code {implementation} satisfy
\begin{equation}
N_{\mathrm{phys}}=\Theta(W_{\log}d^2), \qquad T_{\mathrm{phys}}=\Theta(D_{\log}d).
\end{equation}
These relations give the {overhead in physical qubits and the cost in code cycles for} the specified fault-tolerant compilation once the distance $d$ is fixed. The reliability estimate applies to the stated local stochastic noise model {when the physical error rate is below threshold}. For noise models with strong spatial or temporal correlations, non-Markovian bath memory, leakage outside the computational subspace without a leakage-reduction mechanism, coherent systematic faults, adversarial faults, or physical error rates above $p_{\rm th}$, the bound on $p_L(d)$ requires a separate {analysis for that noise model}. The formal diagonal theorem remains unchanged, while the physical probability of observing the intended pointer relation must be assessed under the relevant noise model. The quantities $D_{\log}$ and $W_{\log}$ characterize the logical circuit, $d$ specifies the code distance, and $N_{\rm phys}$ and $T_{\rm phys}$ characterize the {overheads in physical resources}. These quantities are distinct from the primary control variables entering the diagonal actuation step.}

\section{DISCUSSION {AND ARCHITECTURAL IMPLICATIONS}}
\label{sec:Discussion}

Theorem \ref{thm:noprediction-letter} {realizes} Wolpert's limitation on self-predicting systems {within finite time and a single experimental run \cite{Wolpert2008}. This is an obstruction} under strict resource bounds{. Polynomial-time compilation and one} experimental run suffice to expose the failure. The limitation does not stem from noise, decoherence, or insufficient computational power{. It stems} from the logical structure of {the protocol.}

{
Concretely, the obstruction arises when a protocol class satisfies the reflective closure {condition and therefore admits the forbidden} causal path from the {description of the protocol being executed}, through an analysis module, {to an actuator that can affect the pointer during the same run}. The relevant controls {are} feed-forward settings {defined at the protocol level. The} MZI phase $\phi${, for example, provides a physical realization of such a setting} (see Sec.~\ref{sec:godel-safe-mzi-application}){. Fixed Hamiltonian details are} implementation-specific. The accompanying fault-tolerant resource and reliability estimates apply
under the noise assumptions used {to derive those estimates}. The following applications focus on architectures in which this causal path is either present or explicitly blocked by architectural constraints.
}

Before turning to specific applications, a central 
question is whether randomization can evade the 
diagonal obstruction. Consider an \emph{outcome 
selector} $\Sigma$: a deadline-bounded module that 
takes a protocol specification {$\xi={C}$} 
and outputs a definite bit $b\in\{0,1\}$ intended to 
match the single-shot pointer record. One might 
design $\Sigma$ to incorporate randomness, using 
internal classical noise, quantum fluctuations, or 
device-independently certified random bits, hoping 
to evade deterministic diagonalization through 
irreducible stochasticity. Apply 
Theorem~\ref{thm:noprediction-letter} with 
$A_\Sigma(\xi):=\Sigma(\xi)$. The diagonal {construction} 
produces protocol $C_\Sigma$, unitarily embeddable 
via reversible control and uncomputation, whose 
recorded bit satisfies {$X = 1 - \Sigma({C_\Sigma})$}. Crucially, $C_\Sigma$ executes 
$\Sigma$ within the same run, observes the 
\emph{realized} value of $b$ produced in that 
specific instance, and actuates $X:=1-b$. Since 
diagonalization acts on the actual bit generated 
(not its prior distribution), $\Sigma$ fails with 
certainty regardless of the randomness source 
quality, distinguishing this obstruction from 
statistical sampling limitations where failures are 
probabilistic.

This establishes a fundamental boundary{. Unitary} dynamics and the Born rule
algorithmically predict measurement statistics over repeated runs{. However},
no deadline-bounded procedure acting solely on 
protocol specifications can universally determine 
individual single-shot records whose frequencies 
realize those statistics. The limitation is logical 
rather than statistical, persisting in the noiseless 
limit and applying to deterministic and randomized 
procedures alike.

{
Analogously to classical undecidability results, we emphasize that the obstruction established here is a strict worst-case limitation. The diagonal argument guarantees the construction of at least one specific, adversarially compiled instance {${\mathcal{C}_A}$} on which a given predictor must fail. It does not imply that predictive control fails on typical or practical inputs. Operationally, a predictor may achieve high accuracy and remain broadly effective across the vast majority of standard, non-self-referential tasks. What the theorem explicitly rules out is the stronger guarantee of \emph{universal} {correctness within the same run over any protocol class satisfying the reflective closure condition}. This distinction is particularly important for real-time quantum error correction{. In that domain, fault tolerance} relies on strict worst-case reliability, and predictable deterministic failures pose a challenge that average-case utility cannot entirely offset. To contextualize this theoretical limit, we now evaluate the extent to which modern control stacks physically satisfy the reflective closure {condition}.


\subsection{The {Reflective} Closure {Condition} in Modern Control Hardware}
\label{sec:Closure}

The practical applications of Theorem~\ref{thm:noprediction-letter} concern {protocol classes satisfying the reflective closure condition. The forbidden causal path routes the result of} self-analysis {performed during the run into an actuator that can affect the pointer}, rather than {routing} ordinary measurement feedback:
\begin{equation}
\label{eq:unsafe-self-analysis-path}
{\mathcal{C}_r}
\longrightarrow
A({\mathcal{C}_r})
\longrightarrow
u_r
\longrightarrow
X_r .
\end{equation}
Here, {${\mathcal{C}_r}$} is the finite description of the protocol executed in run $r$, $A$ is a deadline-bounded predictor, certifier, or validation module, $u_r$ is a {control setting applied during run $r$ that can affect the pointer}, and $X_r$ is the pointer record. The architectural implication is that self-analysis of the current run may be recorded or used to configure later runs, but must not be routed {to an actuator that can modify} the predicted or certified record {during that run}.

{The operational relevance} of this vulnerability depends on whether physical controllers can instantiate {Eq.~\eqref{eq:unsafe-self-analysis-path}} within finite resources. Modern platforms possess finite memory and instruction sets{. They increasingly support} access to stored data {during a run} and low-latency feedback{. These features make the reflective closure condition} physically meaningful within finite resource regimes.

Historically, quantum controllers operated as static arbitrary waveform generators (AWGs), compiling and uploading pulse sequences for passive playback. Such open-loop architectures do not satisfy the {reflective closure condition}, as they cannot introspectively evaluate their own stored program during execution.

By contrast, contemporary quantum control stacks have shifted toward dynamic, highly integrated designs combining programmable logic, shared-memory architectures, tightly coupled data paths, and real-time feedback pathways. Hardware platforms based on FPGA and SoC technologies (including systems such as QICK~\cite{Stefanazzi2022QICK}, RISC-Q~\cite{Liu2026RISCQ}, OPX programmed with QUA~\cite{QUA26}, and ZQCS~\cite{ZQCS26}) support architectures with low-latency classical control and conditional execution during a run. These features make it possible, in principle, to implement control flows in which stored protocol descriptions or control parameters can be accessed and processed during execution. At the software level, control frameworks and pulse languages further exploit these capabilities, compiling analysis routines that act on stored descriptions or control parameters and route the resulting signals into feed-forward operations {during the same run}.

{Satisfaction of the reflective closure condition} is not automatic{; it depends on} specific software configurations. The diagonal construction requires an explicit engineering choice to expose the protocol description to the predictor input during the same run. Because the {construction} produces a specific finite instance {${\mathcal{C}_A}$}, its realization requires only that the hardware store and evaluate this description within its execution window. Under these conditions, the obstruction can be physically instantiated.

\subsection{G\"{o}del-Safe Architectures}
\label{sec:GodelSafeDefinition}

To prevent the deterministic failure induced by diagonalization, programmable quantum platforms require explicit structural constraints. Having established that modern shared-memory architectures {can admit the forbidden causal path}, we formalize this requirement {as the G\"odel-safe condition}.

\subsubsection{Formal definition: the G\"odel-safe {condition}}
\label{sec:godel-safe-architecture}

\begin{definition}[\textbf{G\"odel-safe architecture}\footnote{
{The term ``G\"odel-safe'' is used to emphasize the role of self-reference and diagonal constructions, in the spirit of G\"odel's original formulation of semantic self-reference, where formal systems encode statements about their own structure.
}}]
\label{def:godelsafe}

Let $U_r^{(X)}$ denote the set of {control variables available during run $r$} that can
causally influence the pointer record $X_r$. A controller {has a}
\textbf{G\"odel-safe {architecture}} with respect to $X_r$ if {it satisfies the
\textbf{G\"odel-safe condition}: within run $r$,} there is no causal path from the
{description of the protocol being executed} to those variables:
\begin{equation}
\label{eq:godel-safe-condition}
{\mathcal C_r}
\not\leadsto
U_r^{(X)}
\qquad
\text{within run } r .
\end{equation}
Equivalently, the forbidden {causal path} is
\begin{equation}
\label{eq:godel-forbidden-pattern}
{\mathcal C_r}
\to
A({\mathcal C_r})
\to
U_r^{(X)}
\to
X_r .
\end{equation}
\end{definition}

Ordinary real-time feedback remains admissible under {the G\"odel-safe} condition. In
particular, feedback from physical measurement data has the allowed form
\begin{equation}
\label{eq:ordinary-feedback-allowed}
Y_{<t}^{\rm sensor}
\to
u_t
\to
X_r ,
\end{equation}
where $Y_{<t}^{\rm sensor}$ denotes sensor data available before the
control action. The {G\"odel-safe} condition also permits {analysis of the current protocol during the run} to be
routed to non-actuating destinations or later runs:
\begin{equation}
\label{eq:future-run-allowed}
{\mathcal C_r}
\to
A({\mathcal C_r})
\to
\text{log, audit, or future run}.
\end{equation}

The resulting restriction is therefore precise: {the forbidden causal path} is excluded, while standard measurement feed-forward {and fixed decoding remain admissible, as do validation after a run and adaptation between runs}. {Put plainly, the G\"odel-safe condition does not prohibit feedback in general. Sensor data and syndrome measurements may still control the current run. What it excludes is the forbidden causal path from the description of the protocol being executed, through the analysis module, to a control setting that can affect the pointer during that run.}

The {G\"odel-safe condition in} Eq.~\eqref{eq:godel-safe-condition} excludes the hypotheses required for Theorem~\ref{thm:noprediction-letter}. A G\"odel-safe architecture therefore renders the diagonal {construction} physically inadmissible, restoring the possibility of universal predictive consistency at the cost of expressiveness analyzed below. {This condition applies} only to predictors in the sense defined here; standard syndrome-based QEC feedback remains admissible.

The separation required by Eq.~\eqref{eq:godel-safe-condition} can be realized via physically disjoint memories (as in Harvard architectures), via memory-protection mechanisms in unified designs, or via {restrictions enforced by the compiler} that prevent access to the protocol description. Blocking the forbidden {causal path}~\eqref{eq:godel-forbidden-pattern} can be enforced by
hardware wiring (no physical connection between predictor output and control
registers upstream of the pointer) or by bus-arbitration logic that blocks
predictor-to-control transactions affecting $X_r$. By enforcing this strict
separation, the {forbidden causal path is physically blocked}: the predictor $A$
can no longer compute $A({\mathcal C_A})$ inside the same run, thereby
preventing physical realization of the {self-application at the fixed point that
the diagonal construction requires}.

{The G\"odel-safe} condition is sufficient, but not strictly necessary{, for preventing the diagonal construction. Alternative enforcement mechanisms may also prevent the construction. Examples include static analysis performed by the compiler and} restrictive type systems that forbid self-reference algorithmically{. Any such mechanism must prevent every} compiled instance $\mathcal{C}_A$ from evaluating the predictor $A$ on its own description within a single run. {Hardware enforcement, however,} is theoretically unambiguous, directly implementable in contemporary FPGA control stacks, and less dependent on software correctness.

\subsubsection{Worked application: Violation of the G\"odel-safe condition in real-time phase control}
\label{sec:godel-safe-mzi-application}

To make the definition concrete, we examine the MZI implementation of Sec.~\ref{sec:mzi-letter} through the lens of the G\"odel-safe {condition}. This example {deliberately violates that condition}: it realizes the minimal instance of the forbidden {causal path}~(\ref{eq:godel-forbidden-pattern}), thereby exposing the vulnerability in physical hardware.

The controlled quantum system is the single-photon path qubit, and the {interferometric phase is the control variable that can affect the pointer}. During a run, the controller computes
\begin{equation}
\label{eq:app-b-def}
b={P_A^{(0)}(}{\mathcal{C}_r}),
\end{equation}
computes the diagonal bit
\begin{equation}
\label{eq:app-bdiag-def}
b_{{\mathrm{diag}}}=1-b,
\end{equation}
and sets the phase according to
\begin{equation}
\label{eq:app-phase-command}
\phi=\pi b_{{\mathrm{diag}}}.
\end{equation}
In the ideal limit of the MZI, with the detector assignment $D_0\mapsto X=0$ and $D_1\mapsto X=1$, the optical pointer satisfies
\begin{equation}
\label{eq:app-X-bdiag}
X_r=b_{{\mathrm{diag}}}.
\end{equation}
Combining Eqs.~\eqref{eq:app-b-def}--\eqref{eq:app-X-bdiag} yields
\begin{equation}
\label{eq:app-diagonal-relation}
X_r={1-P_A^{(0)}(}{\mathcal{C}_r}).
\end{equation}
For the fixed-point run $\mathcal{C}_A$, this becomes
\begin{equation}
\label{eq:app-fixed-point-relation}
X_{\mathcal{C}_A}
=
{1-P_A^{(0)}(}{\mathcal{C}_A}),
\end{equation}
which is the diagonal contradiction realized as a single-photon measurement {controlled by phase}. Thus, the MZI is not introduced as a new hardware primitive; it serves as a minimal physical example showing how a {bit produced by self-analysis can determine a setting of the quantum control during the same run}.

This example also identifies the practical restriction. A controller that permits the causal chain
\begin{equation}
\label{eq:app-forbidden-mzi-path}
{\mathcal{C}_r}
\to
A({\mathcal{C}_r})
\to
\phi_r
\to
X_r
\end{equation}
satisfies the reflective closure condition used by the theorem. A G\"odel-safe {architecture blocks this dependency within the same run}. The output of {$A({\mathcal{C}_r})$} may be archived, {verified after the run}, or used to update future protocols{. It is excluded, however,} from setting the current phase $\phi_r$ when $\phi_r$ is causally upstream of the pointer $X_r$.

\subsection{Real-Time Quantum Error Correction: A Case Study {of the G\"odel-Safe Condition}}
\label{sec:QEC_CaseStudy}

{Real-time} quantum error correction (QEC) {provides a concrete case study of the obstruction's architectural implications}. Fault-tolerant architectures require classical control stacks capable of processing syndrome data and executing feed-forward corrections within the coherence time of the physical qubits \cite{Google2023SurfaceCode,GoogleQuantumAI2024QEC}. As control algorithms grow in sophistication, the boundary between data processing and instruction scheduling becomes increasingly porous{. Prospective} next-generation QEC implementations {therefore come} within the scope of the limits established by Theorem~\ref{thm:noprediction-letter}.

\subsubsection{{Current} real-time QEC {and the G\"odel-safe condition}}
\label{sec:GS-QEC}

Existing {demonstrations of} real-time QEC {that implement} low-latency decoding, such as Refs.~\cite{Sivak2023RTQEC,GoogleQuantumAI2024QEC}, are consistent with {the} G\"{o}del-safe {condition} in their current implementations. Their decoders operate as fixed classical routines on syndrome streams, without {access during a run} to their executable protocol descriptions. The decoder receives a history of syndrome bits $S_{<t}$, applies a precompiled decoding graph or lookup table, and outputs a correction instruction $u_t$. The resulting causal path
\begin{equation}
\label{eq:qec-safe-pattern}
S_{<t} \to \text{Decoder}(S_{<t}) \to u_t
\end{equation}
matches the allowed pattern of Eq.~\eqref{eq:ordinary-feedback-allowed}, as the input consists exclusively of physical measurement data rather than a representation of the running protocol.

Current QEC designs therefore remain within standard feedback topologies{. They} do not expose their own executable descriptions to analysis within the decoding loop {during a run} and thus do not instantiate the reflective closure condition of Theorem~\ref{thm:noprediction-letter}. However, the drive toward minimal latency has already led to highly integrated architectures that bypass intermediary CPUs{. Distributed} RISC-Q platforms{, for example, demonstrate decoding and feedback loops with an} end-to-end {latency of} 446 ns \cite{Liu2026RISCQ}. As physical routing latencies approach their fundamental limits, prospective next-generation designs may attempt to optimize correction by making the decoder \emph{instruction-aware}{. Such a decoder would} access information about its own execution state or pending control instructions{, such as} pipeline status or instruction queues{,} during the run. Such architectures would {make control decisions during the current run depend} on representations of the active control process. In certain configurations, this would approach the forbidden {causal path} of Eq.~\eqref{eq:godel-forbidden-pattern}, in which the control specification influences {actuation during the current run}. In this regime, the diagonal construction applies{. It implies} the existence of specific protocol instances that provoke deterministic prediction failures. This structural vulnerability is in tension with the strict worst-case reliability guarantees required for fault-tolerant QEC operation.

\subsubsection{Hardware enforcement and expressiveness trade-offs}
\label{sec:qec-trade-offs}

{Future QEC architectures that add introspection within the decoding loop} for latency optimization or meta-learning {must enforce the} G\"odel-safe {condition in hardware. The required separation is} ${\mathcal{C}_r} \not\leadsto U_r^{(X)}${, as stated in} Eq.~\eqref{eq:godel-safe-condition}{. It} can be enforced through physically disjoint decoder data paths from instruction memory (Harvard architectures), memory-protection units that block access to protocol description regions, or bus-arbitration logic preventing predictor-to-correction transactions within the same run.

This enforcement imposes concrete expressiveness costs{. It excludes} control capabilities that rely on introspection of the control process {during the same run}. First, \emph{meta-learning {within the decoding loop}} is sacrificed{. The} decoder cannot update its own model parameters based on {an audit, performed in real time,} of its computational load, timing, or control structure. Second, \emph{dynamic latency optimization} is constrained{. The} decoder cannot inspect its own execution queue or pipeline state to adapt scheduling decisions based on {congestion observed in real time}. Third, \emph{introspective adaptation {at the instruction level}} is excluded {because} it would require access to representations of the active protocol during the same run. Instead, QEC protocols must rely on conservatively compiled worst-case bounds for classical execution time.

These trade-offs {restrict} the space of admissible control policies{. In} unrestricted architectures, control actions may depend {on both} physical measurement data {and the description ${\mathcal{C}_r}$ of the protocol being executed. A} G\"odel-safe architecture {instead excludes the forbidden causal path from that description to an actuator that can affect the pointer during the same run}. As in classical computability and compiler theory (Sec.~\ref{Sec:computabilty}), this {restriction on control flows recovers} verifiable guarantees at the expense of flexibility. In the present setting, the consequence is architectural{. Any extension toward adaptive decoding that is instruction aware} must explicitly enforce the G\"odel-safe {condition. It must accept} a controlled and selective reduction in expressiveness to preserve the worst-case reliability required for fault-tolerant operation.

\subsection{Extensions to Other Autonomous Domains}

Beyond the primary architectural vulnerabilities in real-time quantum error correction, the obstruction established by Theorem~\ref{thm:noprediction-letter} applies to any quantum system whose control stack {admits the forbidden causal path}. The following sections highlight both concrete application domains and their broader computational implications.

}

\subsubsection{Adversarial validation {of predictors and digital twins}}
\label{sec:adversarial-validation}

Testing and validation of quantum hardware and classical simulators present fundamental challenges. Established methods include quantum volume benchmarking \cite{Cross2019}, randomized benchmarking for gate fidelities \cite{Emerson2005}, and application-oriented assessments for the NISQ era \cite{Preskill2018}. A particularly relevant context concerns the validation of \textit{Digital Twins}, classical simulators that emulate the precise noisy dynamics of specific quantum processors \cite{Jaschke2024DigitalTwin,Muller2025DigitalTwin}. 

Our compiler $A\mapsto\mathcal{C}_A$ provides a systematic adversarial harness{.} {{It applies to} modules that assert universal {correctness for single-shot outcomes over a protocol class satisfying the reflective closure condition}. Such a module may be described as a predictor, certifier, validator, or digital-twin model{. The} theorem applies only {to a module that} acts as a deadline-bounded map{. In addition, the module's output must drive an actuator that can affect the pointer during the current run. Under these conditions}, executing $\mathcal{C}_A$ produces a concrete self-referential instance on which the first-bit prediction deterministically fails.}

{This statement does not concern standard digital-twin applications, which typically predict distributions, calibrated dynamics, or {performance in future runs}. Those uses fall outside the diagonal failure mode unless their outputs are routed {to actuation that can affect the pointer during the same run}. The practical role of the theorem is therefore to test and delimit universal self-prediction claims, leaving ordinary statistical validation intact. This adversarial test complements existing benchmarking methods by probing logical self-consistency rather than statistical fidelity{. It} can be integrated into testing pipelines for quantum control software as a verification tool for G\"odel-safe {architectures}.}

\subsubsection{Device-Independent Quantum Certification}

Device-independent (DI) certification \cite{Mayers2004, Acin2007, Supic2020, Gocanin2022} provides statistical guarantees that quantum processes are genuine by observing violations of Bell inequalities{. It does so} without trusting the internal hardware. While DI protocols resist noise and adversarial spoofing, their statistical guarantees implicitly assume fixed causal structures. Our theorem indicates that a DI certifier can be subject to deterministic failure if its {architecture admits the forbidden causal path}. Consequently, the structural limitations imposed by a G\"{o}del-safe architecture provide a necessary causal foundation that complements statistical device independence.

\subsubsection{Single-Shot Quantum Machine Learning}

Single-shot QML frameworks \cite{RecioArmengol2024SSQML,Liu2025YOMO} aim to extract task-relevant inferences directly from single-shot measurement records{. They thereby eliminate} the need for statistical averaging over repeated runs. Theorem~\ref{thm:noprediction-letter} operates in precisely this regime: it does not rely on asymptotics or IID assumptions. Crucially, even in the noiseless, physically deterministic limit, any bounded-time predictor fails on its own self-referential instance. The limitation is purely computational and not a consequence of quantum indeterminacy. It establishes that no deadline-bounded procedure can universally deduce single-shot outcomes from protocol specifications alone. Therefore, single-shot QML pipelines must {exclude the forbidden causal path} to maintain verifiable single-shot guarantees.

\subsubsection{{Limits Induced by Self-Reference across Computational Frameworks}}

{The architectural vulnerabilities identified across these quantum domains are not isolated hardware artifacts, but operational manifestations of a universal structural boundary. This limit on predictive universality is conceptually aligned with the work of Prokopenko \emph{et al.}~\cite{prokopenko2019self}, who emphasized the role of self-reference and related diagonal structures in generating undecidable behavior in complex dynamical systems. While their analysis focuses on emergent dynamics in generalized systems, the present work considers the same underlying mechanism in the context of physically implemented control.}

{Building on this generalized perspective, 
recent work by Wolpert~\cite{wolpert2024implications} applies Kleene’s second recursion theorem to the cosmological simulation hypothesis{. The analysis demonstrates} the formal possibility of self-simulation in computational universes obeying the physical Church--Turing thesis{. It also invokes} Rice’s theorem to derive associated undecidability results. While our construction defines a finite-time operational constraint for quantum control hardware and Wolpert’s analysis addresses abstract computational universes, both rely on the same recursion-theoretic machinery. This convergence across domains highlights a common principle: unrestricted self-reference imposes fundamental limits on universal predictive guarantees.}

{Finally, the structural role of self-reference in limiting consistent inference is conceptually related to Frauchiger and Renner's no-go theorem for self-referential reasoning in quantum theory~\cite{frauchiger2018quantum}. However, we emphasize that the obstruction identified here is strictly computational and does not rely on quantum measurement postulates or Wigner's friend-type scenarios. Instead, it arises from the logical topology of self-reference, illustrating that universal predictive consistency breaks down when a system attempts to formally analyze its own execution logic.}

\subsection{{Insights from computability theory for G\"{o}del-safe {architectures}}}
\label{Sec:computabilty}

{{Classical computability theory provides a useful design analogy.} Turing's halting problem shows that no procedure can decide, in general, whether an arbitrary program halts. Rice's theorem provides a powerful generalization: every non-trivial semantic property of a program is undecidable~\cite{Rice1953}. {Its} proof proceeds by reduction from the halting problem{. This reduction illustrates} how an apparently isolated limitation extends to broad classes of questions about program behavior. These results do not restrict the class of computable functions, since modern programming languages remain Turing-universal{. Instead, they} limit what can be {verified} prior to execution. Compiler theory addresses this tension through type systems and other conservative analyses{. These methods} restrict program structure, making specific properties decidable and ensuring static guarantees{. They thereby trade} expressive flexibility for reliability~\cite{Pierce2002TAPL,Cousot1977,Aho2006Compilers}.}

{{The} same diagonal mechanism {produces an analogous constraint here. In both settings}, the limitation stems from attempting to evaluate properties of a system within its own formal description. In our setting, this becomes a restriction on admissible {control strategies used during the same run. The restriction excludes the forbidden causal path}. G\"odel-safe architectures implement {it} by blocking precisely {that path}. The result is a trade-off between certifiability and expressiveness{. Achieving} universal predictive consistency requires structural constraints. The progression from the halting problem to Rice's theorem and ultimately to constraints in compiler design illustrates how abstract impossibility results can give rise to concrete design principles. Our contribution follows the same trajectory{. Starting} from Wolpert's abstract impossibility, we develop a constructive {realization under finite resources and derive} explicit resource bounds and expressiveness trade-offs{. We then identify the forbidden causal path} underlying the obstruction and formulate its exclusion as {the G\"odel-safe condition implemented} by G\"odel-safe architectures.}


\section{CONCLUSIONS}
\label{sec:Conclusions}

We {provide} an operational no-go theorem for universal
deterministic single-shot prediction under finite resources. {It gives} a
constructive realization of {Wolpert's limitation} on self-predicting systems
{\cite{Wolpert2008}. For} any fixed bounded-time predictor{,} the {compiler based on the recursion theorem
constructs a protocol that runs in finite time. In a single shot, that protocol
produces a pointer record that contradicts the} forecast on that {self-referential instance
\cite{Rogers1967}. Reversible}, unitarily embeddable control and Bennett
uncomputation {make the obstruction} realizable as a {laboratory protocol rather
than merely an abstract computability statement
\cite{Bennett1973,Bennett1989}. The result complements} undecidability phenomena
in quantum dynamics {\cite{PeralesEceiza2025PhysRep}. It remains} logically
independent of Bell and Kochen--Specker constraints
\cite{Bell1964,KochenSpecker1967}.

{Beyond the constructive result itself, the analysis identifies a concrete architectural rule for self-analyzing quantum technologies. The {description of the current protocol} may be inspected, logged, audited, or used to configure later runs{. When a universal} prediction guarantee is asserted {for that same run, however, the forbidden causal path from that description to an actuator that can affect the pointer during the run must be blocked. This restriction} preserves ordinary measurement feedback and fixed {decoders operating in real time. It removes only} the dependency exploited by the diagonal construction. {We introduced and formally defined \textit{G\"odel-safe architectures} to} make this design principle explicit in autonomous frameworks{. Real-time} Quantum Error Correction {is a central example because} worst-case reliability matters operationally. In this form, the result makes explicit a trade-off between control expressiveness and certifiability{. This trade-off mirrors} the structural restrictions adopted in classical compiler design. Enforcing these structural constraints is therefore an increasingly relevant engineering consideration for maintaining logical consistency as {systems for quantum control} become more autonomous and introspective.}

Several {theoretical} directions remain open. It is natural to delineate maximal non-self-referential families of experiments on which deadline-bounded prediction can be complete{. A related goal is} to quantify trade-offs between universality, time bounds, and admissible adaptivity. {This line of inquiry may clarify whether instruction-aware predictors can remain statistically useful in average-case scenarios despite their guaranteed worst-case failure.} One can also extend the analysis to approximate or distributional predictors, where success is probabilistic rather than {deterministic}{. Such an extension should} seek {precise} reliability bounds under diagonal adversarial compilation. {On the practical engineering side, we have identified specific {ways to enforce the} G\"odel-safe {condition in QEC. Extending} and optimizing analogous constraints for other advanced domains, such as digital twins and device-independent certification, remains a challenge.}

{As programmable {loops for quantum control} evolve toward greater computational expressiveness, the limits of self-reference transition from theoretical bounds to concrete constraints on architecture and system design. Enforcing these architectural boundaries will become increasingly important for the reliable operation of future autonomous quantum technologies.}

\section*{Data availability statement}
No new data were created or analysed in this study. Data sharing is not applicable to this article.

\ack{We thank Francesco Marino for helpful discussions and valuable comments on the manuscript.}

\clearpage


\renewcommand{\HALT}{\mathrm{HALT}_{\mathcal M_{\mathrm{rev}}}}
\renewcommand{\thesection}{S\arabic{section}}
\renewcommand{\theequation}{S\arabic{equation}}
\renewcommand{\thefigure}{S\arabic{figure}}
\renewcommand{\thetable}{S\arabic{table}}

\setcounter{section}{0}
\setcounter{equation}{0}
\setcounter{figure}{0}
\setcounter{table}{0}

\numberwithin{equation}{section}
\numberwithin{figure}{section}
\numberwithin{table}{section}

\begin{center}
{\Large\bfseries Supplementary Material for: Constructive realization of self-referential prediction limits in quantum control: Resource bounds and Gödel-safe architectures\par}
\vspace{1em}
Salman Sajad Wani\par
Qatar Center for Quantum Computing, College of Science and Engineering, Hamad Bin Khalifa University, Doha, Qatar\par
\vspace{0.5em}
\'Alvaro Perales-Eceiza\par
Computer Engineering Department, Polytechnic School, Universidad de Alcal\'a, 28805 Madrid, Spain\par
\vspace{0.5em}
Saif Al-Kuwari\par
Qatar Center for Quantum Computing, College of Science and Engineering, Hamad Bin Khalifa University, Doha, Qatar\par
\vspace{0.5em}
Mir Faizal\par
Canadian Quantum Research Center, 204-3002 32 Ave, Vernon, BC V1T 2L7, Canada\par
Irving K.\ Barber School of Arts and Sciences, University of British Columbia Okanagan, Kelowna, BC V1V 1V7, Canada\par
Department of Mathematical Sciences, Durham University, Upper Mountjoy, Stockton Road, Durham DH1 3LE, UK\par
Faculty of Sciences, Hasselt University, Agoralaan Gebouw D, Diepenbeek, 3590, Belgium\par
\end{center}

\section{Computational model and clocked predictors}
\label{sec:SMI}

This section specifies the reversible computational framework used throughout
the diagonal compilation. Its purpose is to fix, once and for all, the notions
of program description length, reversible time evolution, halting under a
deadline, and predictors whose outputs are guaranteed to appear by a prescribed clock time. All subsequent constructions in the Supplementary Material rely on this framework{. They also} inherit their resource accounting from it. Reversible logic admits a unitary embedding up to the terminal pointer readout{. This property provides} a convenient bookkeeping interface for later physical embedding; see, e.g., Bennett’s foundational work on reversible computation and uncomputation \cite{Bennett1973,Bennett1989}.

\runin{Codes and lengths}
We begin by fixing coding conventions used to quantify program size and to
synchronize time budgets with description length.
Let $\{0,1\}^*$ denote the set of finite binary strings.
For any finitely specified object $X$ (program, circuit, protocol), let
${X}\in\{0,1\}^*$ be a fixed binary code for $X$, and let $|{X}|$
denote its bit-length. This convention provides a uniform way to measure description complexity and to
define asymptotic resource bounds as functions of code length. 

\begin{definition}[Universal reversible machine]
\label{def:SMI-URM}
We fix a universal reversible Turing machine $\mathcal M_{\mathrm{rev}}$, as in
Bennett \cite{Bennett1973,Bennett1989}.
A program is a binary string interpreted as code for
$\mathcal M_{\mathrm{rev}}$.
The running time of a program is defined as the number of applications of the
reversible transition rule of $\mathcal M_{\mathrm{rev}}$ starting from the
standard initial configuration.
\end{definition}
This definition fixes a canonical notion of reversible time evolution against
which all later clock bounds and step counts are measured. Universality ensures
that all computable reversible procedures can be expressed within this single
model.

\runin{Scheduled halting convention}
To define a halting event compatible with reversibility, each program carries a
one-bit halt flag $h$ in its configuration, and a clock register stores the step
index $t$. The event ``halting at step $t$'' means that the predicate $h=1$ holds
at clock value $t$, together with a first-time condition defined in
Sec.~\ref{sec:SMII}. This scheduled halting convention preserves bijectivity of the step-to-step
evolution while enabling a verifiable halting signal {at the prescribed deadline}. The
explicit separation between the halt flag and the clock register is essential
for defining bounded reversible simulations in later sections.

\begin{definition}[Time-constructible budgets]
\label{def:SMI-budget}
A function $\tau:\mathbb N\to\mathbb N$ is time-constructible if there
exists a program $T_\tau$ that, on unary input $1^k$, runs for exactly $\tau(k)$
steps and outputs $\tau(k)$ in binary on a designated output tape.
\end{definition}
Time-constructibility ensures that the computational budget associated with a
given code length can itself be generated, verified, and aligned with the
reversible clock model. This property allows deadlines to be incorporated into
the computation without introducing external timing assumptions.

\begin{definition}[Clocked deterministic predictors]
\label{def:SMI-predictor}
Fix a time-constructible budget $\tau$.
A clocked deterministic predictor is a total computable map
$A:\{0,1\}^*\to\{0,1\}^{m}$ encoded as a program for $\mathcal M_{\mathrm{rev}}$
with code length $k:=|{A}|$, such that for every input
$\xi\in\{0,1\}^*$ the execution reaches a designated output register whose
contents equal $A(\xi)$ by step $\tau(k)$, and subsequently enters a reversible
idle cycle that leaves that register unchanged. We write $P_A(\xi):=A(\xi)$ and denote the first output bit by
$P_A^{(0)}(\xi)$.
In the diagonal construction we set $m:=k+1$ (one diagonal bit plus a $k$-bit
payload).
\end{definition}
This definition formalizes the notion of a predictor whose output is guaranteed
to be available by a fixed, length-dependent deadline and remains stable
thereafter. The explicit idle cycle ensures that output stability is compatible
with reversibility and can be checked at arbitrary later times.

\begin{definition}[Uniform payload bijections]
\label{def:SMI-pi}
Fix a family of bijections $\pi_k:\{0,1\}^k\to\{0,1\}^k$ with uniformly computable
inverses $\pi_k^{-1}$. Equivalently, there exists a single constant-size program
which, on inputs of length $k$, implements $\pi_k$ (or $\pi_k^{-1}$) using $k$
inferred from the input length.
Example: $\pi_k=\mathrm{id}$ for all $k$.
\end{definition}
These bijections are used to uniformly permute payload bits without affecting
the diagonal logic. Their role is purely informational: they preserve
invertibility while allowing the observed output record to encode auxiliary
data in a uniform and length-respecting manner.

\section{Bounded reversible simulation and the first-time halting bit}
\label{sec:SMII}

This section defines a uniform reversible gadget {for detecting first-time halting.
Given} a reversible program $Q$ and a step bound $t$, {the gadget} computes a single
classical bit indicating whether $Q$ raises its halt flag for the first time
exactly at step $t$. The purpose of this construction is to isolate the minimal
amount of bounded simulation required for the diagonal compilation. {The} gadget
serves as the (D) block in the diagonal construction{. It also} provides the
step-bounded simulation primitive used throughout the Supplementary Material.

\runin{Step-indexed halt flag}
For a program $Q$, let $h_s(Q)\in\{0,1\}$ denote the value of its halt flag after
$s$ reversible steps, starting from the standard initial configuration. This
notation allows the halting behavior of $Q$ to be treated as a time-indexed
Boolean sequence, which makes it possible to express ``first-time halting'' as a
single predicate depending on the entire execution history up to time $t$.

\begin{definition}[Bounded first-time halting predicate]
\label{def:SMII-halt}
For $({Q},t)\in\{0,1\}^*\times\mathbb N$, define
\begin{equation}
\label{eq:SMII-HALTdef}
\HALT({Q},t)
\;:=\;
h_t(Q)\,\prod_{s=0}^{t-1}\bigl(1-h_s(Q)\bigr)
\;\in\;\{0,1\}.
\end{equation}
Equivalently, $\HALT({Q},t)=1$ holds exactly when $h_t(Q)=1$ and
$h_s(Q)=0$ for every $s<t$.
\end{definition}
Equation~\eqref{eq:SMII-HALTdef} encodes the logical requirement that the halt flag
is raised at step $t$ and has not been raised at any earlier step. The product
over $s<t$ enforces the first-time condition, while the factor $h_t(Q)$ ensures
that halting occurs at the designated step.

\begin{lemma}[Uniform reversible circuit for $\HALT({Q},t)$]
\label{lem:SMII-FY}
There exists a uniform compiler that, on input $({Q},t)$, outputs a
logically reversible circuit $\mathcal C_{{Q},t}$ with the following
properties.

(i) On an all-zero workspace and a fresh output bit $H$, the circuit simulates
the first $t$ steps of $Q$, computes $H=\HALT({Q},t)$, and then reverses its
internal computation so that immediately before terminal measurement all
registers except $H$ return to their initial values.

(ii) Let $D$ and $W$ denote the logical depth and width of
$\mathcal C_{{Q},t}$.
There exist machine-dependent constants such that
\begin{equation}
\label{eq:SMII-DW}
D \;=\; O\!\bigl(t\log t + |{Q}|\bigr),
\qquad
W \;=\; O\!\bigl(t + |{Q}|\bigr).
\end{equation}
\end{lemma}
This lemma guarantees that the first-time halting predicate defined in
Definition~\ref{def:SMII-halt} can be evaluated reversibly with explicit and
controlled resource overhead. The construction ensures that no irreversible
information about intermediate configurations remains after the output bit is
produced.

\begin{proof}
The circuit maintains a reversible simulation register for the configuration of
$Q$, a binary clock register of $\lceil\log_2(t{+}1)\rceil$ bits, a one-bit
accumulator $e$, and a sequence of ancillas $g_1,\dots,g_t$ initialized to $0$.
For each step $s=1,\dots,t$, it applies one reversible transition of $Q$ and
increments the clock. The clock register ensures that the simulation is aligned
with the step index used in the halting predicate.
Let $h$ be the halt flag in the current configuration. At each step, the circuit
computes $g_s:=h(1-e)$ into the ancilla $g_s$ by a reversible controlled write,
then updates $e:=e\oplus g_s$. By construction, the accumulator $e$ records
whether a halting event has already occurred. As a result, at most one $g_s$ can
equal $1$, and $g_s=1$ holds exactly at the first step where $h=1$.
Consequently, $g_t=1$ holds if and only if the first occurrence of $h=1$ is at
step $t$, which matches the predicate defined in
Eq.~\eqref{eq:SMII-HALTdef}. The circuit copies $g_t$ into the designated output
bit $H$ and then reverses the entire computation to clean the workspace, ensuring
that all ancillas and simulation registers are restored to their initial states.
The width bound follows from the workspace required to store the simulated
configuration over $t$ steps together with the program description overhead.
The depth bound includes a factor of $\log t$ arising from binary clock updates
and conditional control operations.
\end{proof}

\section{Budget-aligned diagonal compilation and the contradiction bit}
\label{sec:SMIII}

This section implements the diagonal step under explicit and verifiable
computational budgets. Fix a clocked deterministic predictor $A$
(Definition~\ref{def:SMI-predictor}) with code length
$k:=|{A}|$ and associated deadline $\tau(k)$. The goal is to construct a
finitely specified reversible experiment $\mathcal C_A$ whose recorded first
pointer bit is complementary to the predictor’s first output bit evaluated on
the experiment’s own description. All time bounds are aligned so that every subroutine completes before its
scheduled deadline. In particular, compilation overhead, predictor execution,
and the halting decision are separated into disjoint time windows.

\runin{Compilation overhead}
We fix a compilation pipeline that builds the bounded-halting gadget
$\mathcal C_{{Q},t}$ from Lemma~\ref{lem:SMII-FY} and marshals the resulting
circuit code as an input string to the predictor $A$. Let
${g_{\mathrm{comp}}}:\mathbb N^2\to\mathbb N$ be a time-constructible polynomial such that, given
$(x,k)$ and $t\ge 1$, the pipeline completes within ${g_{\mathrm{comp}}}(k,|x|)$ reversible steps.
Concretely, within this time bound the pipeline constructs
$s={\mathcal C_{{x},t}}$ and writes $s$ onto the input interface used to
invoke $A(s)$. This polynomial upper bound isolates compilation overhead from
predictor runtime and allows both contributions to be budgeted explicitly.

\begin{definition}[Diagonal deadline]
\label{def:SMIII-t0}
For $(x,k)\in\{0,1\}^*\times\mathbb N$, define
\begin{equation}
\label{eq:SMIII-t0}
t_0(x,k)\;:=\;{g_{\mathrm{comp}}}(k,|x|)\; +\;\tau(k)\; +\;1.
\end{equation}
\end{definition}
The additional unit step in Eq.~\eqref{eq:SMIII-t0} ensures a clean separation
between the completion of the predictor output and the scheduled halting event,
which is essential for enforcing a reversible {halting policy tied to the prescribed deadline}.

\begin{lemma}[Aligned diagonal program $D_A$]
\label{lem:SMIII-DA}
There exists a reversible program $D_A$ with code
$x^\ast:={D_A}$ such that, with $t_0:=t_0(x^\ast,k)$ and
$s:={\mathcal C_{{x^\ast},t_0}}$, the execution produces
\begin{equation}
\label{eq:SMIII-bdef}
b\;:=\;P_A^{(0)}(s)
\end{equation}
by step $t_0-1$, and the halt flag satisfies
\begin{equation}
\label{eq:SMIII-haltpolicy}
h_s(D_A)=0\ \text{for all}\ s<t_0,
\qquad
h_{t_0}(D_A)=1 \iff b=0.
\end{equation}
Consequently,
\begin{equation}
\label{eq:SMIII-bdiag}
\HALT({D_A},t_0)=1-P_A^{(0)}(s).
\end{equation}
\end{lemma}
This lemma encodes the diagonal logic itself: the scheduled halting decision of
$D_A$ is explicitly tied to the negation of the predictor’s first output bit,
while strictly respecting the global deadline $t_0$.

\begin{proof}
Define a total computable transformer $\Theta$ that maps $(x,{A})$ to the
code of a reversible machine $M_x$ with the following behavior. First, it
computes $t_0=t_0(x,k)$. Next, it constructs and marshals
$s={\mathcal C_{{x},t_0}}$ within ${g_{\mathrm{comp}}}(k,|x|)$ steps. It then runs $A(s)$
within $\tau(k)$ steps, obtaining $b=P_A^{(0)}(s)$ by step $t_0-1$.
Finally, $M_x$ follows a timed halting policy: the halt flag remains $0$ up to
time $t_0-1$ and equals $1$ at time $t_0$ exactly when $b=0$. This policy enforces
the diagonal condition without introducing irreversible control flow.
By Kleene’s recursion theorem \cite{Rogers1967}, there exists a fixed point
$x^\ast$ such that the resulting machine equals $M_{x^\ast}$. Defining $D_A$ by this fixed point yields
$\HALT({D_A},t_0)=1$ exactly when $b=0$, which gives
Eq.~\eqref{eq:SMIII-bdiag}.
\end{proof}

\begin{definition}[Compiled experiment $\mathcal C_A$ and pointer record]
\label{def:SMIII-CA}
Let $x^\ast:={D_A}$ and $t_0:=t_0(x^\ast,k)$ be as above. Define $\mathcal C_A$
as the reversible circuit that performs:

(i) (D) Run $\mathcal C_{{D_A},t_0}$ to compute the diagonal bit
\begin{equation}
\label{eq:SMIII-bdiag-def}
b_{\mathrm{diag}}:=\HALT({D_A},t_0)
\end{equation}
into a fresh pointer qubit.

(ii) (P) Write the payload
\begin{equation}
\label{eq:SMIII-Rdef}
R:=\pi_k({A})\in\{0,1\}^k
\end{equation}
into a clean $k$-qubit pointer register using reversible bit operations.

(iii) (U) Apply the inverse internal computation so that immediately before
terminal measurement only the $m=k+1$ pointer qubits carry nontrivial content.

Immediately before measurement the pointer register equals the computational
basis state
$\ket{b_{\mathrm{diag}}}\otimes\ket{R}$. The recorded classical outcome string is
\begin{equation}
\label{eq:SMIII-YA}
Y_A:=\bigl(b_{\mathrm{diag}},R\bigr)\in\{0,1\}^{k+1}.
\end{equation}
\end{definition}
This definition makes explicit which degrees of freedom are measured and which
are uncomputed, ensuring that the observed record is both classical and
information-preserving.

\begin{lemma}[Payload invertibility and injectivity on length-$k$ predictors]
\label{lem:SM-injective-payload}
Fix $k\in\mathbb N$ and let $A$ be any predictor with $|{A}|=k$. In the
payload variant of Definition~\ref{def:SMIII-CA}, the predictor code is recovered
from the observed record by ${A}=\pi_k^{-1}(R)$. Hence the map
$\mathcal F_k:{A}\mapsto Y_A$ is injective on the set of length-$k$ codes.
\end{lemma}

\begin{proof}
By definition, the last $k$ bits of $Y_A$ equal $R=\pi_k({A})$. Applying the
inverse bijection yields $\pi_k^{-1}(R)={A}$.
\end{proof}

\begin{remark}[Role of the payload]
\label{rem:SM-payload-role}
The diagonal contradiction uses only the first pointer bit
$b_{\mathrm{diag}}$. The payload $R$ provides an explicit,
information-preserving record of the predictor description in the observed
output. Its role is auxiliary: it enforces injectivity without affecting the
diagonal logic itself.
\end{remark}

\begin{lemma}[Diagonal contradiction on the first output bit]
\label{lem:SMIII-contradict}
Let $\mathcal C_A$ and $Y_A$ be as in Definition~\ref{def:SMIII-CA}. Then
\begin{equation}
\label{eq:SMIII-contradict-eq}
b_{\mathrm{diag}}
\;=\;
1-P_A^{(0)}\!\bigl({\mathcal C_A}\bigr),
\end{equation}
hence $A({\mathcal C_A})\neq Y_A$.
\end{lemma}

\begin{proof}
Definition~\ref{def:SMIII-CA} gives
$b_{\mathrm{diag}}=\HALT({D_A},t_0)$. Lemma~\ref{lem:SMIII-DA} gives
$\HALT({D_A},t_0)=1-P_A^{(0)}(s)$ with
$s={\mathcal C_{{x^\ast},t_0}}$. The fixed point
$x^\ast={D_A}$ is chosen so that the pipeline-generated code string $s$
equals the actual description ${\mathcal C_A}$ of the compiled circuit
using $(D_A,t_0)$ in its (D) block.
Substituting yields Eq.~\eqref{eq:SMIII-contradict-eq}. Since the first bit of
$Y_A$ equals $b_{\mathrm{diag}}$ while the first bit of
$A({\mathcal C_A})$ equals
$P_A^{(0)}({\mathcal C_A})$, the full strings differ.
\end{proof}

\section{Mach--Zehnder realization: timing and one-shot error envelopes}
\label{sec:SMIV}

{
This section describes a feed-forward implementation of the diagonal bit
$b_{\mathrm{diag}}$ using a single-photon Mach--Zehnder interferometer (MZI).
The section has two purposes: first, to demonstrate that the diagonal bit can
be implemented as a binary interferometric pointer {whose phase is controlled, using}
standard linear optics; and second, to provide explicit causal timing
conditions and one-shot error envelopes that bound the probability of an
incorrect pointer outcome in a single experimental run.}

{\runin\runin{Ideal interferometric pointer}}
Consider an MZI with balanced $50{:}50$ beam splitters BS$_1$, BS$_2$, a phase
shifter on the upper arm, and threshold detectors $D_0,D_1$ at the outputs. The
path degree of freedom is treated as a qubit with computational basis
$\{\ket{0},\ket{1}\}$, where the phase shifter applies the unitary transformation
$\ket{1}\mapsto e^{i\phi}\ket{1}$. Under the standard convention \cite{BornWolf,SalehTeich}, a single photon entering
BS$_1$ prepares the path superposition
$\ket{\psi_{\mathrm{path}}}=(\ket{0}+\ket{1})/\sqrt 2$. After recombination at
BS$_2$, the ideal click probabilities at the output detectors satisfy
\cite{BornWolf,SalehTeich}
\begin{equation}
\label{eq:SMIV-fringes}
\Pr(D_0\mid\phi)=\cos^2\!\frac{\phi}{2},
\qquad
\Pr(D_1\mid\phi)=\sin^2\!\frac{\phi}{2}.
\end{equation}
We encode the classical detector outcome by setting $X=0$ for a click at $D_0$
and $X=1$ for a click at $D_1$. {Equation~\eqref{eq:SMIV-fringes} specifies the ideal phase-to-port probability
map from the phase setting $\phi$ to the detection probabilities for the
classical binary pointer.}

\runin{Actuation rule}
The diagonal controller outputs the bit $b_{\mathrm{diag}}\in\{0,1\}$
(Sec.~\ref{sec:SMIII}) and commands the phase shifter according to
\begin{equation}
\label{eq:SMIV-phase-rule}
\phi=\pi\,b_{\mathrm{diag}}.
\end{equation}
In the ideal model, substituting Eq.~\eqref{eq:SMIV-phase-rule} into
Eq.~\eqref{eq:SMIV-fringes} yields $X=b_{\mathrm{diag}}$ with unit probability.
{Thus, in the absence of noise and timing violations, the MZI serves as a
faithful binary interferometric pointer {whose phase is controlled by} the diagonal bit.}

\runin{Timing constraint}
We now specify the causal timing condition required for correct feed-forward
operation. Let $T_{\mathrm{ctrl}}$ denote the worst-case time required to marshal
${\mathcal C_A}$, evaluate the predictor call, compute
$b_{\mathrm{diag}}$, and issue the phase command. Let $T_\pi$ denote the
phase-modulator settling time, and let $T_{\mathrm{flight}}$ denote the photon
time-of-flight through the controlled arm, including any deliberate delay. A sufficient causal-alignment condition is
\begin{equation}
\label{eq:SMIV-timing}
T_{\mathrm{flight}} \ \ge\ T_{\mathrm{ctrl}}+T_\pi+\Delta,
\end{equation}
where $\Delta>0$ is a fixed safety margin. This inequality ensures that the phase
setting determined by the diagonal computation is applied before the photon
reaches the modulated arm. In a digital controller model with step time $T_{\mathrm{step}}$, one may bound
$T_{\mathrm{ctrl}}\le t_0\,T_{\mathrm{step}}$ using the diagonal deadline $t_0$
from Definition~\ref{def:SMIII-t0}. This connects the causal timing requirement
directly to the computational budget.

\runin{Noise model and envelopes}
We next introduce a simple noise model sufficient to bound the one-shot error
probability. We assume matched detectors with end-to-end detection efficiency
$\eta\in(0,1]$, a per-port background click probability $d\ll 1$ per detection
gate, and a postselection rule that discards double-click events. Imperfect interference is summarized by a fringe visibility $V\in[0,1]$ and a
phase-setting error $\epsilon$ about the target values
$\phi\in\{0,\pi\}$, with $|\epsilon|\le\delta\phi$. Under these assumptions, a
standard visibility model replaces Eq.~\eqref{eq:SMIV-fringes} by
\cite{SalehTeich}
\begin{equation}
\label{eq:SMIV-vis}
\Pr(D_0\mid\phi)=\frac{1+V\cos\phi}{2},
\qquad
\Pr(D_1\mid\phi)=\frac{1-V\cos\phi}{2}.
\end{equation}
Equation~\eqref{eq:SMIV-vis} captures the combined effect of imperfect mode
overlap and residual dephasing on the output statistics.

\begin{lemma}[Signal wrong-port probability]
\label{lem:SMIV-q}
Under Eq.~\eqref{eq:SMIV-vis} and $|\epsilon|\le\delta\phi$,
\begin{equation}
\label{eq:SMIV-qdef}
\Pr[X\neq b_{\mathrm{diag}}\mid \text{signal click}]
\ \le\
q
\ :=\
\frac{1-V}{2}+\sin^2\!\frac{\delta\phi}{2}.
\end{equation}
\end{lemma}

\begin{proof}
With target phase $\phi=\pi b_{\mathrm{diag}}$ and realized phase
$\phi+\epsilon$, the probability of a wrong-port signal click equals
$(1-V\cos\epsilon)/2$. Using $\cos\epsilon=1-2\sin^2(\epsilon/2)$ yields
$(1-V)/2+V\sin^2(\epsilon/2)\le (1-V)/2+\sin^2(\delta\phi/2)$, which establishes
the bound.
\end{proof}

Let $\mathrm{acc}$ denote the event that exactly one detector clicks in the
detection gate.

\begin{lemma}[Postselected one-shot envelope]
\label{lem:SMIV-post}
To first order in $d$,
\begin{equation}
\label{eq:SMIV-post}
\Pr[X\neq b_{\mathrm{diag}}\mid \mathrm{acc}]
\ \le\
q+\frac{1-\eta}{\eta}\,d.
\end{equation}
\end{lemma}

\begin{proof}
Accepted events {have two possible sources. A} registered signal click {occurs} with
probability $\eta+O(d)${, and its} error probability is bounded by $q$
(Lemma~\ref{lem:SMIV-q}){. Alternatively, a signal event is missed (with}
probability $1-\eta$) {and} a single background click {occurs}. To first order in
$d$, the latter contributes at most $(1-\eta)d$ to the numerator and $\eta$ to
the denominator, yielding Eq.~\eqref{eq:SMIV-post}.
\end{proof}

\begin{lemma}[Unconditioned one-shot envelope]
\label{lem:SMIV-unpost}
To first order in $d$,
\begin{equation}
\label{eq:SMIV-unpost}
\Pr[X\neq b_{\mathrm{diag}}]
\ \le\
q+d.
\end{equation}
\end{lemma}

\begin{proof}
An error arises either from optical misrouting of a signal click, which is
bounded by $q$, or from a wrong-port background click, which contributes at most
$d$ to first order. Summing these contributions yields the stated bound.
\end{proof}
\section{Resource scaling and a fault-tolerant embedding}
\label{sec:SMV}

This section tracks the logical and physical resources required to implement the
compiled experiment $\mathcal C_A$. Its purpose is to make explicit how circuit
depth, width, and reliability scale with the predictor code length $k$, and to
record a standard fault-tolerant embedding that justifies the main-text claims
regarding polynomial overhead and scalable reliability.

\runin{Logical depth and width}
Let $x^\ast:={D_A}$ and $t_0:=t_0(x^\ast,k)$ be as defined in
Sec.~\ref{sec:SMIII}. The (D) block of $\mathcal C_A$ is the bounded reversible
simulation circuit $\mathcal C_{{D_A},t_0}$ from
Lemma~\ref{lem:SMII-FY}. Let $D_{\log}$ and $W_{\log}$ denote the logical depth and
width of $\mathcal C_A$, excluding the $m=k+1$ pointer qubits that carry the
measured classical record. Lemma~\ref{lem:SMII-FY} directly yields the following asymptotic bounds on the
logical resources:
\begin{equation}
\label{eq:SMV-DW-general}
D_{\log}
=\;O\!\bigl(t_0\log t_0+|x^\ast|+k\bigr),
\qquad
W_{\log}
=\;O\!\bigl(t_0+|x^\ast|+k\bigr).
\end{equation}
These expressions separate the contribution of the bounded simulation horizon
$t_0$ from the additive overhead associated with program descriptions and payload
handling. No other sources of asymptotic growth appear in the compiled circuit.

\begin{lemma}[Fixed-point code length]
\label{lem:SMV-fp}
For the fixed diagonalization pipeline of Sec.~\ref{sec:SMIII}, there exists a
constant $c_{\mathrm{fp}}$ independent of $A$ such that
\begin{equation}
\label{eq:SMV-fp}
|x^\ast|\le k+c_{\mathrm{fp}}.
\end{equation}
\end{lemma}

\begin{proof}
The fixed point consists of a constant-size wrapper implementing the fixed
transformer and self-reference mechanism together with a hardwired copy of
${A}$. This introduces only an additive constant overhead beyond $k$
\cite{Rogers1967}.
\end{proof}
This bound ensures that the self-referential construction does not introduce
superlinear growth in description length.

\begin{corollary}[Polynomial deadline under polynomial predictor budgets]
\label{cor:SMV-t0}
Define ${\bar g_{\mathrm{comp}}}(k):={g_{\mathrm{comp}}}(k,k+c_{\mathrm{fp}})$. Then
\begin{equation}
\label{eq:SMV-t0}
t_0 \le \tau(k)+{\bar g_{\mathrm{comp}}}(k)+1.
\end{equation}
If $\tau(k)$ is polynomial, then $t_0$, $D_{\log}$, and $W_{\log}$ are polynomial
functions of $k$.
\end{corollary}

\begin{proof}
Definition~\ref{def:SMIII-t0} gives
$t_0={g_{\mathrm{comp}}}(k,|x^\ast|)+\tau(k)+1$, and
Lemma~\ref{lem:SMV-fp} yields $|x^\ast|\le k+c_{\mathrm{fp}}$. Substituting gives
${g_{\mathrm{comp}}}(k,|x^\ast|)\le {\bar g_{\mathrm{comp}}}(k)$, which implies the stated bound.
\end{proof}
This corollary makes explicit that, under polynomial predictor budgets, the
entire diagonal experiment remains within polynomial logical resources.

{
\runin{Fault-tolerant embedding (surface code)}
{We describe a standard fault-tolerant embedding{. Under the}
local-stochastic noise model stated here, {this embedding} is sufficient to
execute $\mathcal C_A$ with arbitrarily small logical failure probability
{when physical error rates are below threshold}. We assume a local-stochastic noise model with physical two-qubit error rate
$p$ satisfying $p<p_{\mathrm{th}}$, where $p_{\mathrm{th}}$ is the
surface-code threshold for the specified layout and decoder. Let $d$ denote
the code distance and $p_L(d)$ the logical failure probability per logical
location.}
}

\begin{lemma}[Logical error per location and a union bound]
\label{lem:SMV-surface}
For $p<p_{\mathrm{th}}$ there exist a constant $c_0=O(1)$ such that
\begin{equation}
\label{eq:SMV-pL}
p_L(d)\le c_0\Bigl(\frac{p}{p_{\mathrm{th}}}\Bigr)^{(d+1)/2}.
\end{equation}
Here $c_0$ and $p_{\mathrm{th}}$ are determined by the decoder, layout, and noise
model.
For a compiled computation with $N_{\mathrm{loc}}$ logical locations,
\begin{equation}
\label{eq:SMV-union}
P_{\mathrm{logical}}\le N_{\mathrm{loc}}\,p_L(d).
\end{equation}
\end{lemma}

\begin{proof}
Equation~\eqref{eq:SMV-pL} expresses the standard exponential suppression of
logical errors with increasing code distance under $p<p_{\mathrm{th}}$ in
phenomenological noise models \cite{Fowler2012i}, with $c_0$ absorbing decoder-
and layout-dependent constants. The bound~\eqref{eq:SMV-union} follows from a
union bound over logical locations.
\end{proof}

\begin{proposition}[Exponential suppression with linear distance]
\label{prop:SMV-exp}
Fix $p<p_{\mathrm{th}}$ and choose $d:=a k+b$ with constants $a>0$ and $b\ge 0$.
Assume $\tau(k)$ is polynomial so that
$N_{\mathrm{loc}}=\mathrm{poly}(k)$ under the conservative estimate
$N_{\mathrm{loc}}=\alpha D_{\log}W_{\log}$.
Then there exist constants $\gamma>0$ and $k_0$ such that for all $k\ge k_0$,
\begin{equation}
\label{eq:SMV-expbound}
P_{\mathrm{logical}}\le 2^{-\gamma k}.
\end{equation}
\end{proposition}

\begin{proof}
Combining Eq.~\eqref{eq:SMV-union} with Eq.~\eqref{eq:SMV-pL}, the factor
$(p/p_{\mathrm{th}})^{(d+1)/2}$ scales as $2^{-\Theta(k)}$ when $d=\Theta(k)$.
Since $N_{\mathrm{loc}}$ grows only polynomially in $k$, the exponential
suppression dominates for sufficiently large $k$, yielding the stated bound.
\end{proof}
This proposition shows that scalable reliability is achieved with only linear
growth of code distance in the predictor length.

\runin{Physical resources}
A standard surface-code compilation yields the physical scaling
$N_{\mathrm{phys}}=\Theta(W_{\log}d^2)$ physical qubits and
$T_{\mathrm{phys}}=\Theta(D_{\log}d)$ code cycles for patch-based layouts and
lattice-surgery style primitives \cite{Fowler2012i} . These scalings complete the
resource accounting for the fault-tolerant implementation of $\mathcal C_A$.

\bibliographystyle{apsrev4-2}
\bibliography{name}

\end{document}